\documentclass[manuscript]{acmart}

\usepackage{mathtools}
\usepackage{tikz}
\usepackage{adjustbox}
\usepackage{soul}
\usepackage{booktabs}
\usepackage[flushleft]{threeparttable}
\usepackage[ruled,vlined,linesnumbered]{algorithm2e}
\usetikzlibrary{quantikz}
\usepackage{cleveref}

\theoremstyle{acmplain}

\newtheorem{theorem}{Theorem}
\newtheorem{proposition}{Proposition}
\newtheorem{lemma}{Lemma}
\newtheorem{corollary}{Corollary}
\theoremstyle{acmdefinition}

\newtheorem{definition}{Definition}
\newtheorem{remark}{Remark}

\graphicspath{{./}{figures/}{../figures/}}

\providecommand{\ket}[1]{}\renewcommand{\ket}[1]{\lvert #1 \rangle}
\providecommand{\bra}[1]{}\renewcommand{\bra}[1]{\langle #1 \rvert}

\crefname{figure}{Figure}{Figures}
\Crefname{figure}{Figure}{Figures}
\crefname{table}{Table}{Tables}
\Crefname{table}{Table}{Tables}
\crefname{algocf}{Algorithm}{Algorithms}
\Crefname{algocf}{Algorithm}{Algorithms}

\usepackage[final,commandnameprefix=ifneeded]{changes}
\AtBeginDocument{\hypersetup{colorlinks=true, linkcolor=black, citecolor=black,
                            urlcolor=black, filecolor=black, menucolor=black}}
\setaddedmarkup{{\color{blue!75!black}#1}}
\setdeletedmarkup{{\color{red!65!black}[\,#1\,]}}
\usepackage{enumitem}
\setlist{itemsep=1pt, topsep=2pt, parsep=0pt}
\definechangesauthor[name={revision}, color=blue]{rev}
\definechangesauthor[name={meeting 2}, color=orange]{m2}

\newcommand{\numVerGrid}{41}
\newcommand{\numVerFormulaErr}{6.7\times10^{-16}}
\newcommand{\numVerCascadeErr}{3.9\times10^{-16}}

\newcommand{\numPowRmax}{3}
\newcommand{\numCostPZero}{0.599}
\newcommand{\numCostMargin}{0.0993}
\newcommand{\numCostShots}{187}
\newcommand{\numEnvQiskit}{2.5.2}
\newcommand{\numEnvAer}{0.17.2}
\newcommand{\numEnvPython}{3.11}
\newcommand{\numThrLevelLo}{0.19}
\newcommand{\numThrLevelHi}{0.81}
\newcommand{\numThrSep}{0.311}
\newcommand{\numThrRmax}{10.10}
\newcommand{\numThrRwrap}{11}
\newcommand{\numThrWrap}{3.42}
\newcommand{\numThrPowSep}{1.87}
\newcommand{\numOveRegShots}{202}

\newcommand{\numOveEstShots}{522}
\newcommand{\numOvePlainShots}{73\,781}

\newcommand{\numOveRatioEstShots}{2.59}

\newcommand{\numOveEstMaxPower}{32}
\newcommand{\numOveRatioPlain}{365}
\newcommand{\numOveSlopePlain}{-1.99}
\newcommand{\numOveSlopeEstShots}{-0.48}
\newcommand{\numOveSlopeReg}{-0.30}
\newcommand{\numOveAccFloor}{0.97}
\newcommand{\numOveAccTarget}{0.95}

\begin{document}

\replaced[id=m2]{\title{A Quantum Phase-based Comparator}}{\title{QuAC: A
Width-One Quantum Comparator for Amplitude-Encoded Values}}

\author{Alessandro Berti}
\orcid{0000-0001-9144-9572}
\affiliation{%
  \department{Department of Computer Science}
  \institution{University of Pisa}
  \city{Pisa}
  \country{Italy}}
\email{alessandro.berti1994@gmail.com}

\author{Alessandro Poggiali}
\orcid{0000-0002-1591-7925}
\affiliation{%
  \department{Department of Computer Science}
  \institution{University of Pisa}
  \city{Pisa}
  \country{Italy}}
\email{alessandro.poggiali@di.unipi.it}

\renewcommand{\shortauthors}{\footnotesize \textsc{\textbf{Preprint}}}

\begin{abstract}
Quantum comparators decide the order of two operands. They rely on reversible logic acting on basis-encoded integers, so their width grows with the precision. We introduce the Quantum Phase-based Comparator (QPC), which compares
two values carried in relative phases instead. The circuit places the two
phases, entered with opposite signs, between a pair of Hadamard gates, and a
fixed offset centers the interference, so that the probability of measuring
zero falls below one half exactly when the first phase is the smaller. We give
two implementations. With \emph{hard-coded values}, the two values are written
into the parameters of two phase gates, and the comparison costs one qubit and
five gates. With the \emph{register-driven cascade}, the values are drawn from
two $t$-qubit registers through binary-weighted controlled-phase gates, at a
cost of one qubit beyond the registers and depth linear in $t$; since the
cascade induces each phase linearly from the register content, one circuit
handles a superposition of operand pairs. The readout yields a biased coin
rather than a definite bit, and we quantify the shots that a decision takes at
a given phase separation. Scaling both phases by an integer widens the
decision margin at no cost in width or depth, and an adaptive doubling schedule
turns this amplification into a shot count logarithmic in the inverse
separation. 
\end{abstract}

\replaced[id=m2]{}{}

\maketitle

\section{Introduction}\label{sec:intro}

The predicate $x<y$ is a primitive in several application domains of quantum
algorithms. It thresholds pixel values in quantum image
processing~\cite{xia2019}, triggers the conditional subtraction of quantum
modular arithmetic~\cite{yuan2023}, and serves as the oracle that minimum- and
maximum-finding searches call $\mathcal{O}(\sqrt{M})$
times~\cite{durr1996,ahuja1999}. Almost every quantum comparator in the
literature evaluates this predicate on integers held in registers: reversible
comparator trees~\cite{thapliyal2010,vudadha2012}, bit-string
comparators~\cite{oliveira2007,shahzad2023} and Fourier-basis
subtraction~\cite{yuan2023,sahin2020} place the operands on $\mathcal{O}(t)$
qubits, where $t$ is the operand bit width, and return a definite bit. Such a
comparator needs a binary encoding of each value before it can start, so the
precision of the operands fixes the width of the comparison.

This work presents the Quantum Phase-based Comparator (QPC), which takes its
values as phases. The circuit writes the two phases $\varphi_x$ and
$-\varphi_y$ onto a single qubit held between two Hadamard gates, in the layout
of the Hadamard test~\cite{ekert2002} or of a Mach--Zehnder interferometer; the
second Hadamard gate turns the accumulated phase difference into an amplitude
difference between $\ket{0}$ and $\ket{1}$. A fixed $P(-\pi/2)$ gate shifts the
interference curve so that the outcome probability equals one half exactly when
the two phases agree, and this shift is what makes the interferometer a
comparator: the side of one half on which the probability falls gives the
order. The phases enter through commuting diagonal gates, and the two
instantiations differ only in how the circuit supplies them. In the first,
\emph{hard-coded values}, a classical computation supplies the phases, and the
comparison costs one qubit and five gates in total; this is the width-one
instantiation. In the second, the \emph{register-driven cascade}, two
$t$-qubit registers supply them through a binary-weighted cascade of
controlled-phase gates, at a cost of one qubit beyond those registers and depth
$\mathcal{O}(t)$; by linearity, this circuit compares a superposition of
operand pairs.
The output of QPC is statistical. A single shot returns a biased coin whose
bias carries the order, not a definite bit, and two operands arbitrarily close
to each other need arbitrarily many shots to separate. We quantify that
cost and then reduce it: scaling both phases by an integer widens the margin at
no cost in width or depth, and an adaptive doubling schedule turns this
amplification into a shot count logarithmic, rather than quadratic, in the
inverse separation.

\paragraph*{Contributions.}
This paper makes the following contributions.
\begin{itemize}
\item We prove that a single-qubit, constant-depth circuit decides the order of
two encoded phases, and that any strictly increasing encoding whose phases span
less than $\pi$ thereby decides the order of the values.
\item We characterize the readout as a statistical decision rule and derive its
shot complexity from the decision margin.
\item We extend the comparator to values held in registers, possibly in
superposition, and prove a rounding window for the finite-precision encoding.
We further prove that scaling the phases by an integer amplifies the margin at
no cost in width or depth, and that an adaptive schedule turns this
amplification into a logarithmic shot count.
\item We implement the comparator in Qiskit~\cite{qiskit2024} and release the
implementation, together with the verification code and the experiment
scripts, in a public
repository.\footnote{\url{https://github.com/AlessandroPoggiali/Quantum_Phase-based_Comparator}}
\end{itemize}

The manuscript is organized as follows. \Cref{sec:related} places QPC among reversible comparators and single-qubit
interference readouts. \Cref{sec:prelim} fixes notation, states the decision
problem, introduces the phase encoding and its scale, and recalls the binary-weighted cascade of controlled
rotations. \Cref{sec:input-model} presents the
comparator: the phase interferometer and its correctness,
the shot cost of the statistical readout, and the two
circuits that realize the primitive, with the rounding window of the
register-driven one. \Cref{sec:accuracy} shows that
scaling the phases by an integer widens the margin at no cost in width or
depth and turns this amplification into a logarithmic
shot count through the adaptive schedule.
\Cref{sec:exp-setup} validates the circuits by simulation, measures the
adaptive readout against plain sampling and the estimate-both baseline, and
applies the comparator to thresholding an image held in superposition.
\Cref{sec:conclusion} concludes.

\section{Related Work}\label{sec:related}

The encoding of the operands divides quantum comparators into two families:
reversible logic compares basis-encoded integers, while interference reads
values held in phases or amplitudes. \Cref{tab:resources} sets both families
side by side with the two supply modes of the phase-based family developed
here, along four axes: width, measured in qubits; depth; gate count; and output
type.

\paragraph*{Comparators for basis-encoded integers.}
This family holds its operands in registers of $\mathcal{O}(t)$ qubits and
compares them with reversible logic. Comparator trees decide order and equality
bitwise~\cite{thapliyal2010,vudadha2012}; bit-string designs flag the relations
into ancillas with Toffoli and CNOT logic~\cite{oliveira2007,shahzad2023}; and
Fourier-basis designs subtract by phase rotations and read the sign of the
difference from a single ancilla~\cite{yuan2023,sahin2020}. The comparators of
Vandaele~\cite{vandaele2026} use $\Theta(t)$ gates and $\Theta(\log_2 t)$ depth
and attain a provably minimal qubit count over Clifford${+}$Toffoli, which
shows that exact basis-encoded comparison needs $\Omega(t)$ width. These
designs remain the exact tool whenever the operands already sit in the
computational basis and the algorithm needs a definite bit.

A second quantum route estimates the operands instead of comparing them
directly. Amplitude estimation~\cite{brassard2002} and its QFT-free
refinements~\cite{suzuki2020,grinko2021} recover an amplitude with a number of
applications of the preparation circuit that scales inversely with the target
error, and Kitaev's phase estimation~\cite{kitaev1995} reads a phase through
powers of the underlying unitary. Estimating both operands and comparing the
estimates classically gives the \emph{estimate-both} baseline against which we
set our readout. The baseline is naive by design: it reconstructs both values
where the order relation alone would suffice.

\paragraph*{Single-qubit interference.}
The mechanism used here is the standard single-qubit interference primitive, a
phase between two Hadamard gates in the layout of the Hadamard
test~\cite{ekert2002}. Closest in mechanism is~\cite{ohno2026}, which encodes
real values as angles and reads a trigonometric function of them through a
Hadamard test; it estimates a continuous similarity scalar, whereas we decide a
thresholded order, so the two answer different questions. Single-qubit readout
also appears in quantum classifiers~\cite{park2021,larose2020,schuld2017},
which classify one input against a learned boundary rather than comparing two
encoded inputs.

Deciding the order rather than reconstructing the values is what lets us build
the comparator around a single interference qubit, and the results of this
paper follow from that choice: a finite-precision guarantee on the decision
(\Cref{prop:window}), a margin that widens at no cost in width or depth
(\Cref{prop:powering}), and a register-driven form that carries the whole
construction over to operand pairs held in superposition (\Cref{sec:register}).

\begin{table}[t]
\centering
\begin{threeparttable}
\caption{Quantum comparison as a resource design space. For the basis-encoded
designs $t$ is the operand width; for our register extension it is the precision
with which the register represents the value, and $\varepsilon$ is the target
additive error of an estimate. Gate counts are analytic, each row over the gate
set of its own construction: Clifford${+}$Toffoli for the reversible designs
of~\cite{thapliyal2010, oliveira2007, vandaele2026}, Clifford with controlled
phase rotations for~\cite{yuan2023}, and for ours Hadamard with uncontrolled
phase gates when the values are hard-coded and Hadamard with controlled-phase
gates in the cascade. A \emph{statistical} output is one obtained from the
measurement frequencies of the comparator qubit. Our two rows differ only in how
the circuit receives its operands: the width-one row applies to hard-coded
values, the register-driven row to the cascade.
}
\label{tab:resources}
\small
\setlength{\tabcolsep}{4pt}
\begin{tabular*}{\linewidth}{@{\extracolsep{\fill}}lccccc@{}}
\toprule
\textbf{Construction} & \textbf{Operands} & \textbf{Qubits} & \textbf{Depth} & \textbf{Gates} & \textbf{Output} \\
\midrule
Reversible tree~\cite{thapliyal2010} & integer & $\mathcal{O}(t)$ & $\mathcal{O}(\log_2 t)$ & $\Theta(t)$ & bit \\
Bit-string~\cite{oliveira2007}       & integer & $\mathcal{O}(t)$ & $\mathcal{O}(t)$ & $\Theta(t)$ & bit \\
QFT-based~\cite{yuan2023}            & integer & $2t{+}1$ & $\mathcal{O}(t^2)$ & $\Theta(t^2)$ & bit \\
Ancilla-free~\cite{vandaele2026}     & integer & $2t{+}1$ & $\Theta(\log_2 t)$ & $\Theta(t)$ & bit \\
Estimate-both~\cite{brassard2002,suzuki2020,grinko2021} & amplitude & $\mathcal{O}(t)$ & $\mathcal{O}(1/\varepsilon)$ & n/a & binary values \\
\midrule
\multicolumn{6}{l}{\textbf{QPC (this work)}} \\
\quad Hard-coded values (\Cref{fig:W}) & phase & $\mathbf{1}$ & $\boldsymbol{\mathcal{O}(1)}$ & $\mathbf{5}$ & statistical \\
\quad Register-driven cascade (\Cref{fig:comparator-cascade}) & integer & $2t{+}1$ & $\mathcal{O}(t)$ & $2t{+}3$ & statistical \\
\bottomrule
\end{tabular*}
\end{threeparttable}
\end{table}

\section{Preliminaries}\label{sec:prelim}

We assume basic knowledge of quantum computing and refer to~\cite{nielsen2010}
for an introduction. We write $\ket{\psi}_{\mathrm{name}}^{\mathrm{size}}$,
where the subscript names the register and the superscript gives its size in
qubits, and we omit either when the context makes it clear. Throughout, the
construction relies on the phase gate
$P(\theta)=\bigl[\begin{smallmatrix}1&0\\0&e^{i\theta}\end{smallmatrix}\bigr]$,
and in particular on its additive composition property
$P(\theta_1)P(\theta_2)=P(\theta_1+\theta_2)$. \Cref{def:problem} states the
problem we address.

\begin{definition}[Encoded-order decision]\label{def:problem}
Let $\varphi$ map a value domain $\mathcal{V}\subseteq\mathbb{R}$ into phases
strictly increasingly, so that $x<y$ holds exactly when
$\varphi(x)<\varphi(y)$ for all $x,y\in\mathcal{V}$. Suppose the circuit
receives two values $x,y\in\mathcal{V}$ only as the phases
$\varphi_x=\varphi(x)$ and $\varphi_y=\varphi(y)$, and fix a confidence
$1-\delta$ and a resolution $\eta>0$. The circuit returns one bit, equal to $1$
when $x<y$ and to $0$ otherwise, which is correct with probability at least
$1-\delta$ whenever the phase separation $d=\varphi_x-\varphi_y$ satisfies
$|d|\ge\eta$. Closer pairs are left unresolved, and the tie $x=y$ is excluded.
\end{definition}

The exact comparators of~\Cref{tab:resources}, those that return a definite
bit rather than a statistical one, decide the same order relation at
$\delta=0$ and $\eta=0$.

\subsection{Encodings}\label{sec:encodings}

\emph{Basis encoding} maps a $t$-bit string to the corresponding computational
basis state of a $t$-qubit register; the string $101$ maps to $\ket{101}$.
Throughout this paper, basis encoding stores the phase $\varphi$ of a value in
fixed-point binary over an angle range $A$: the representable phases form the
grid $\{k\,A/2^{t}: k=0,\dots,2^{t}-1\}$ of step $A/2^{t}$, and a register
$\ket{\varphi}_a^t$ holds $\varphi/A$ rounded to the nearest multiple of
$2^{-t}$, saturating at the largest grid point. Rounding therefore displaces a
phase by at most half a step, $A/2^{t+1}$, as long as the phase does not exceed
the largest grid point $A(1-2^{-t})$; above it, saturation displaces it by
more.

The convention is a register spanning the whole circle, $A=2\pi$, but the
encoding used here never needs that much, and an $A$ that is too large wastes
qubits. The comparator consumes phases, so we encode a value $u$ on a data
interval $[a,b]$ through the linear map $\varphi(u)=c\,(u-a)$, anchored at the left endpoint so that $\varphi(a)=0$ and the phases start where the register grid starts. Two conditions fix
$c$. The phases must span less than a half turn, $|c|\,(b-a)<\pi$ (see
\Cref{thm:quac}), and the strict monotonicity of \Cref{def:problem} forces
$c>0$, since a negative scale keeps the span but inverts every decision. Write
$W=c\,(b-a)$ for the width of that span. Shifting both registers by the same
amount changes nothing, because the comparator sees only the difference: the
width of the data interval matters, its position does not.

The width $W$ and the precision then fix the register's angle range $A$. The
largest representable grid point is $A(1-2^{-t})$, so $A$ must satisfy
$A\,(1-2^{-t})\ge W$; otherwise every phase above that point lands on the same
grid point, and the comparator cannot order the values behind them. We take the
tightest such range, $A=W/(1-2^{-t})$, which exceeds the span by exactly one
grid step.

Although $W<\pi$ suffices for correctness, we assume $W<\tfrac{\pi}{2}$
throughout the paper. The margin $\tfrac12|\sin d|$ of \Cref{sec:cost} is increasing in
$|d|$ on $[0,\tfrac{\pi}{2}]$ and decreasing on $[\tfrac{\pi}{2},\pi]$, so
for $W>\tfrac{\pi}{2}$ the pairs with the largest separation have a smaller
margin than pairs with moderate separation and need more shots to resolve,
even though the decision remains correct.

\subsection{A Binary-Weighted Cascade of Controlled Rotations}\label{sec:cascade-prelim}

A single-qubit rotation through a parametric angle decomposes into fixed-angle
rotations controlled by a register that basis-encodes that angle~\cite{mottonen2005,mitarai2019adc,berti2025qram}.  \Cref{def:cascade}
recalls the construction.

\begin{definition}[Cascade of controlled rotations]\label{def:cascade}
Let $a$ be a $t$-qubit register, $b$ a target qubit, and $R(\cdot)$ a
single-qubit rotation obeying $R(\phi_1+\phi_2)=R(\phi_1)R(\phi_2)$. The
weights $\beta_0,\dots,\beta_{t-1}$ are a fixed, data-independent family of
angles, whose values \Cref{lemma:cascade} gives. The cascade
$\mathrm{C}_a R_{\mapsto b}=\prod_{i=0}^{t-1}\mathrm{C}_{a_i}R_{\mapsto b}(\beta_i)$
applies $R(\beta_i)$ to $b$ when the control qubit $a_i$ is in state $\ket{1}$
and the identity otherwise. For clarity, we omit tensor products with the
identity on the remaining register qubits.
\end{definition}

\begin{lemma}[Cascade of controlled rotations equivalence]\label{lemma:cascade}
Fix an \emph{angle range} $A>0$ and let $a=a_0a_1\cdots a_{t-1}$ basis-encode the
angle $\theta=\sum_{i=0}^{t-1}a_i\,\beta_i$ in fixed-point binary
representation, with the fixed, data-independent weights
\[
\beta_i \;=\; \frac{A}{2^{\,i+1}},\qquad a_0 \text{ the most significant bit.}
\]
Then the cascade of \Cref{def:cascade} with these fixed angles reproduces the
single parametric rotation,
\[
\mathrm{C}_a R_{\mapsto b}
=\prod_{i=0}^{t-1}\mathrm{C}_{a_i}R_{\mapsto b}\!\Big(\tfrac{A}{2^{\,i+1}}\Big)
= R_{\mapsto b}(\theta).
\]
\end{lemma}

For $t=3$ and $A=2\pi$, the register $\ket{011}_a$ encodes
$\theta=\tfrac{\pi}{2}+\tfrac{\pi}{4}=\tfrac{3\pi}{4}$, and the cascade fires
the second and third rotations, whose angles sum to exactly $\tfrac{3\pi}{4}$.
Because the weights do not depend on the data, the gate layout is the same for
every input and only the control values change; each cascade has depth
$\mathcal{O}(t)$.

\section{The Quantum Phase-based Comparator}\label{sec:input-model}
This section presents the phase comparator. \Cref{sec:ph_int} states the core
primitive, a phase interferometer that decides the order of two values from the
phases that carry them, whatever mechanism places those phases on the qubit
(\Cref{thm:quac}), and discusses what rescaling the encoding costs in margin
(\Cref{rem:rescaling}). \Cref{sec:cost} quantifies the cost of the statistical
readout by bounding the number of shots that a decision at a given separation
requires.

\Cref{sec:circuits} then gives two circuits that realize the primitive. The
first hard-codes the two values into phase gates and uses one qubit and five
gates; the second loads them from two $t$-qubit registers through a cascade of
controlled-phase gates and also accepts registers in superposition. Both
circuits inherit the results of the first two subsections, which involve a
circuit only through the phases it accumulates.

\subsection{Phase Interferometer}\label{sec:ph_int}
QPC rests on a phase interferometer. A qubit enters in the superposition
$H\ket{0}$, the two values reach it as phases of opposite sign, $\varphi_x$ and
$-\varphi_y$, and a second Hadamard gate converts the accumulated phase into a
bias of the measurement outcome. On its own, this interferometer measures only
how far apart the two phases are: its outcome probability is an even function
of $\varphi_x-\varphi_y$, so swapping the two values leaves it unchanged. A
fixed offset of $-\tfrac{\pi}{2}$ breaks the symmetry. It places the tie at
probability exactly one half, where the response is steepest, and sends the
probability to one side of one half or the other according to which value is
smaller, as long as the two phases differ by less than $\pi$. \Cref{thm:quac}
states this for arbitrary phases, whatever mechanism places them on the qubit.
The readout is therefore statistical: each shot favors the correct outcome,
and \Cref{sec:cost} counts the shots that a reliable decision requires.

\begin{theorem}[Phase Interferometer]\label{thm:quac}
Let $\varphi_x,\varphi_y\in\mathbb{R}$ be two encoded phases, and set
$\phi=\varphi_x-\varphi_y-\tfrac{\pi}{2}$. Measuring $H\,P(\phi)\,H\ket{0}$ in
the computational basis gives
\begin{equation}\label{eq:p0}
P(\ket{0})=\tfrac12+\tfrac12\sin(\varphi_x-\varphi_y),
\end{equation}
and, whenever $|\varphi_x-\varphi_y|<\pi$, $P(\ket{0})<\tfrac12$ holds exactly
when $\varphi_x<\varphi_y$. Consequently, for any strictly increasing encoding
$\varphi$ whose phases span less than $\pi$, the same rule decides $x<y$ for
$\varphi_x=\varphi(x)$ and $\varphi_y=\varphi(y)$.
\end{theorem}

\begin{proof}[Proof (interference calculation)]
The gate $P(\phi)$ is diagonal, so it acts only on the $\ket{1}$ branch, and
the circuit evolves as
\begin{align*}
\ket{0}\;\xrightarrow{\;H\;}\;
  &\tfrac{1}{\sqrt2}\bigl(\ket{0}+\ket{1}\bigr)\\
\xrightarrow{\;P(\phi)\;}\;
  &\tfrac{1}{\sqrt2}\bigl(\ket{0}+e^{i\phi}\ket{1}\bigr)\\
\xrightarrow{\;H\;}\;
  &\tfrac12\bigl[(1+e^{i\phi})\ket{0}+(1-e^{i\phi})\ket{1}\bigr]\\
=\;
  &e^{i\phi/2}\bigl[\cos\tfrac{\phi}{2}\ket{0}-i\sin\tfrac{\phi}{2}\ket{1}\bigr],
\end{align*}
where the last line factors $e^{i\phi/2}$ out of both amplitudes. The global
phase does not affect the measurement, so
\[
P(\ket{0})
=\cos^2\tfrac{\phi}{2}
=\tfrac12\bigl(1+\cos\phi\bigr)
=\tfrac12+\tfrac12\sin(\varphi_x-\varphi_y),
\]
where the last step uses $\cos(\theta-\tfrac{\pi}{2})=\sin\theta$. The offset
$-\tfrac{\pi}{2}$ in $\phi$ is thus what turns the cosine into a sine.

On $(-\pi,\pi)$ the sine vanishes only at $0$ and is negative exactly when its
argument is. Hence, whenever $|\varphi_x-\varphi_y|<\pi$, the probability
equals $\tfrac12$ only at the tie $\varphi_x=\varphi_y$, and
$P(\ket{0})<\tfrac12$ holds exactly when $\varphi_x<\varphi_y$.

Finally, let $\varphi_x=\varphi(x)$ and $\varphi_y=\varphi(y)$. A strictly
increasing encoding gives $\varphi_x<\varphi_y$ exactly when $x<y$, and a span
below $\pi$ keeps $|\varphi_x-\varphi_y|<\pi$, so the same rule decides $x<y$.
\end{proof}

\begin{remark}[Cost of rescaling]\label{rem:rescaling}
\Cref{thm:quac} admits any strictly increasing map of span below $\pi$; the
freedom that matters in practice is rescaling, and it costs little. For a pair
of values $x,y$, the linear encoding gives the phase separation
$d=\varphi_x-\varphi_y=c\,(x-y)$, so by \Cref{eq:p0} the pair has margin
$|P(\ket{0})-\tfrac12|=\tfrac12|\sin d|$. Replacing $c$ by $c/2$ halves the
separation to $d/2$, and the margin drops to $\tfrac12|\sin(d/2)|$. Since
$\sin d=2\sin(d/2)\cos(d/2)$, we have $|\sin(d/2)|\ge\tfrac12|\sin d|$: the
margin loses at most a factor of two at any separation, and the factor tends to
two as $d\to0$. Doubling the width of the data interval at the same phase span
therefore costs at most one halving of the margin.

\end{remark}



\subsection{Cost of the Statistical Readout}\label{sec:cost}

The comparator returns a probability rather than a bit, so a decision costs
shots. The fixed offset $-\tfrac{\pi}{2}$ places the outcome probability at
exactly $\tfrac12$ when the two phases agree, so the order is carried by the
side of $\tfrac12$ on which $P(\ket{0})$ falls, which is the problem of
\Cref{def:problem}. We therefore estimate $P(\ket{0})$ by the frequency
$\widehat{p}_0$ of $\ket{0}$ over $N$ shots and compare it with $\tfrac12$. The
distance $|\widehat{p}_0-\tfrac12|$ is the \emph{empirical margin}; it
estimates the \emph{margin}
\begin{equation}\label{eq:margin}
\Delta(\varphi_x,\varphi_y)\;=\;\bigl|P(\ket{0})-\tfrac12\bigr|\;=\;\tfrac12\,\bigl|\sin(\varphi_x-\varphi_y)\bigr|,
\end{equation}
which determines the number of shots a decision needs. The
margin~\eqref{eq:margin} is a probability offset, whereas the separation
$|\varphi_x-\varphi_y|$ is an angle; near a tie the two are proportional,
$\Delta\approx\tfrac12|\varphi_x-\varphi_y|$, and under $\varphi(u)=c\,(u-a)$ the
separation is exactly $c$ times the value gap.

Each shot is a Bernoulli trial with mean $P(\ket{0})$, and the margin is
bounded away from zero except at the tie $\varphi_x=\varphi_y$. By Hoeffding's
inequality, the empirical frequency over $N$ shots satisfies
\[
\Pr\bigl[\,|\widehat{p}_0-P(\ket{0})|\ge \Delta\,\bigr]\le 2e^{-2N\Delta^2},
\]
so the decision $\widehat{p}_0<\tfrac12$ agrees with $\varphi_x<\varphi_y$ with
probability at least $1-\delta$ as soon as
\begin{equation}\label{eq:shots}
N \;\ge\; \frac{\ln(2/\delta)}{2\,\Delta(\varphi_x,\varphi_y)^2}
   \;=\; \mathcal{O}\!\Big(\frac{1}{\sin^{2}(\varphi_x-\varphi_y)}\,\log_2\tfrac1\delta\Big).
\end{equation}

Away from ties the bound is a constant number of shots: the pair $x=0.8$,
$y=0.6$ gives $P(\ket{0})=\numCostPZero$ and a margin $\Delta=\numCostMargin$,
so at confidence $0.95$ the bound requires
$N\ge\ln(40)/(2\cdot\numCostMargin^{2})\approx\numCostShots$ shots. As the
separation closes, the margin vanishes linearly and the cost grows
quadratically in the inverse separation. In practice one fixes a target
resolution $\eta$, spends the $N$ of~\Cref{eq:shots} at the worst margin that
$\eta$ admits over the whole separation range $[\eta,W]$, namely
$\tfrac12\sin\eta$ when $W\le\tfrac{\pi}{2}$ and $\tfrac12\min(\sin\eta,\sin W)$
otherwise, and returns the side of $\tfrac12$; pairs closer than $\eta$ are
reported unresolved.

\subsection{Circuit implementations}\label{sec:circuits}
This section gives two circuits that implement the phase interferometer of
\Cref{thm:quac}. In the first, \emph{hard-coded values} (\Cref{sec:comparator},
\Cref{fig:W}), a classical computation supplies the two values and writes them
into two phase gates. In the second, the \emph{register-driven cascade}
(\Cref{sec:register}, \Cref{fig:comparator-cascade}), two $t$-qubit registers
supply them, possibly in superposition.

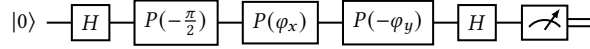
\begin{figure}[t]
\centering
\adjustbox{scale=0.95}{
\begin{quantikz}[row sep=0.3cm, column sep=0.35cm]
\lstick{$\ket{0}$} & \gate{H} & \gate{P(-\tfrac{\pi}{2})} & \gate{P(\varphi_x)} & \gate{P(-\varphi_y)} & \gate{H} & \meter{} &\cw \\
\end{quantikz}}
\caption{Hard-coded values. The circuit realizes \Cref{thm:quac} with the two
values written into the data gates $P(\varphi_x)$ and $P(-\varphi_y)$; the
fixed gate $P(-\tfrac{\pi}{2})$ centers the interference at
$\varphi_x=\varphi_y$, and the rule $P(\ket{0})<\tfrac12$ reads the order.%
}
\Description{A single-qubit circuit: a Hadamard gate, the phase gates P of
minus pi over two, P of phi x, and P of minus phi y, a second Hadamard gate,
and a computational-basis measurement.}
\label{fig:W}
\end{figure}

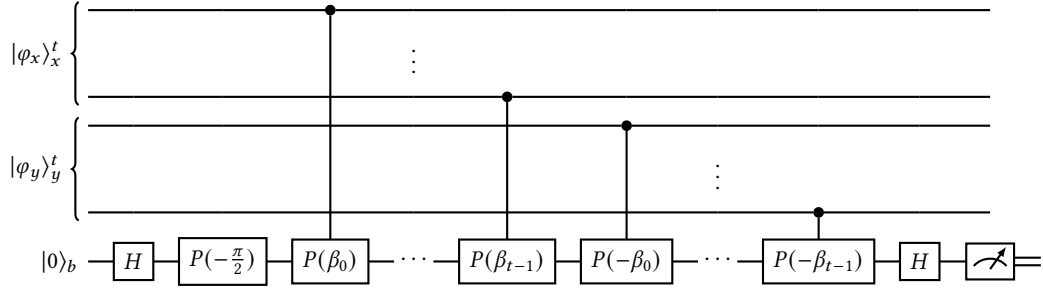
\begin{figure}[t]
\centering
\adjustbox{max width=\linewidth}{%
\begin{quantikz}[row sep=0.24cm, column sep=0.34cm]
\lstick[3]{$\ket{\varphi_x}_x^t$} & \qw & \qw & \ctrl{6} & \qw & \qw & \qw & \qw & \qw & \qw & \qw \\
 & & & & \vdots & & & & & & \\
 & \qw & \qw & \qw & \qw & \ctrl{4} & \qw & \qw & \qw & \qw & \qw \\
\lstick[3]{$\ket{\varphi_y}_y^t$} & \qw & \qw & \qw & \qw & \qw & \ctrl{3} & \qw & \qw & \qw & \qw \\
 & & & & & & & \vdots & & & \\
 & \qw & \qw & \qw & \qw & \qw & \qw & \qw & \ctrl{1} & \qw & \qw \\
\lstick{$\ket{0}_b$} & \gate{H} & \gate{P(-\tfrac{\pi}{2})} & \gate{P(\beta_0)} & \ \cdots\ \qw & \gate{P(\beta_{t-1})} & \gate{P(-\beta_0)} & \ \cdots\ \qw & \gate{P(-\beta_{t-1})} & \gate{H} & \meter{} & \cw
\end{quantikz}}
\caption{Register-driven cascade. The circuit of \Cref{fig:W} with the two data
gates replaced by cascades of controlled-phase gates (\Cref{lemma:cascade}).
Two $t$-qubit registers basis-encode the phases $\varphi_x=c\,(x-a)$ and $\varphi_y=c\,(y-a)$ and drive the cascades with the fixed angles
$\pm\beta_i=\pm A/2^{i+1}$ over the angle range $A$ of \Cref{sec:encodings};
the comparator qubit $\ket{0}_b$ carries the two Hadamard gates and the fixed
$P(-\tfrac{\pi}{2})$ gate, and its measurement decides $x<y$.%
}
\Description{The single-qubit comparator circuit with its two data phase gates
replaced by cascades of controlled phase gates, driven by two t-qubit
registers; the bottom qubit carries the two Hadamard gates, the fixed phase
gate of minus pi over two, and the measurement.}
\label{fig:comparator-cascade}
\end{figure}

\subsubsection*{Hard-coded values}\label{sec:comparator}
A classical machine holds the two values $x$ and $y$, computes $\varphi_x=c\,(x-a)$ and $\varphi_y=c\,(y-a)$, and writes them into the parameters of two phase gates;
no quantum operation precedes the comparator. The circuit (\Cref{fig:W})
applies a Hadamard gate, the fixed offset $P(-\tfrac{\pi}{2})$, the two data
gates $P(\varphi_x)$ and $P(-\varphi_y)$, and a second Hadamard gate, then
measures. The three diagonal gates commute and their angles add, so the
comparator qubit accumulates $\phi=\varphi_x-\varphi_y-\tfrac{\pi}{2}$ and
\Cref{thm:quac} applies. The comparison costs one qubit and five gates at
constant depth, whatever the precision of the values.

\subsubsection*{Register-driven cascade}\label{sec:register}
The hard-coded circuit needs both values on a classical machine. The
register-driven cascade drops that requirement: two quantum registers supply
the operands, and since a register may hold a superposition, one circuit
compares many pairs at once. The thresholding application of
\Cref{sec:app-threshold} relies on this property.

The values reside in two $t$-qubit registers that basis-encode their phases in
the convention of \Cref{sec:encodings}. We apply \Cref{lemma:cascade} with
$R=P$: the data gate $P(\varphi_x)$ becomes a cascade of controlled-phase gates
with fixed angles $+\beta_i$ controlled by $\ket{\varphi_x}_x^t$, and
$P(-\varphi_y)$ becomes a cascade with angles $-\beta_i$ controlled by
$\ket{\varphi_y}_y^t$. \Cref{fig:comparator-cascade} shows the circuit. The
weights $\beta_i=A/2^{i+1}$ do not depend on the data, so the layout is the
same for every input and only the register contents change. Beyond the two
registers, the comparator uses one qubit, two Hadamard gates, the fixed offset
and $2t$ controlled-phase gates: $2t{+}1$ qubits, $2t{+}3$ gates and
$\mathcal{O}(t)$ depth.

Each cascade is diagonal in the computational basis of its register, so it acts
on a superposition one basis state at a time: for every pair
$\ket{\varphi_x}_x\ket{\varphi_y}_y$ the comparator qubit picks up the phase of
that pair and nothing else. What happens next is the algorithm's choice. If it
measures the registers, or an index register entangled with them, together
with the comparator qubit, it reads the decision of a single pair. If it does
not measure, the comparator qubit stays entangled with the registers and
carries the decision of every pair at once, and a downstream routine can
consume this state directly, for instance through amplitude
amplification~\cite{brassard2000quantum} applied to the branches whose
comparator qubit reads $\ket{1}$.

\paragraph*{Finite precision.}
When both phases lie on the grid of \Cref{sec:encodings}, the cascade and the
hard-coded circuit apply the same unitary to the comparator qubit and reach the
same decision. For other values the registers hold rounded phases, and
\Cref{prop:window} bounds the effect of that rounding on the decision.

\begin{proposition}[Rounding window]\label{prop:window}
Let the registers span the angle range $A$, and let the exact phases
$\varphi_x$ and $\varphi_y$ satisfy the band hypothesis
$|\varphi_x-\varphi_y|<\pi-A\,2^{-t}$. Round each phase to the nearest point of
the grid $\{k\,A/2^{t}: k=0,\dots,2^{t}-1\}$ before encoding it. Whenever
\[
|\varphi_x-\varphi_y| \;>\; \frac{A}{2^{t}},
\]
the register-driven cascade returns the same decision as the hard-coded
circuit; disagreements and spurious ties occur only inside the window
$|\varphi_x-\varphi_y|\le A\,2^{-t}$. For the encoding $\varphi(u)=c\,(u-a)$ of
\Cref{sec:encodings} at the range $A=W/(1-2^{-t})$, the band hypothesis holds
for every $t\ge1$ as soon as $W<\tfrac{\pi}{2}$, and the decisions agree
whenever $|x-y|>(b-a)/(2^{t}-1)$.
\end{proposition}

\begin{proof}[Proof (rounding perturbation bound)]
Round-to-nearest on a grid of step $A/2^{t}$ moves each phase by at most
$A/2^{t+1}$, so the encoded difference $\delta_r$ satisfies
$|\delta_r-(\varphi_x-\varphi_y)|\le A/2^{t}$. The separation hypothesis
$|\varphi_x-\varphi_y|>A/2^{t}$ then makes $\delta_r$ nonzero with the sign of
$\varphi_x-\varphi_y$, and the band hypothesis gives
$|\delta_r|<\pi-A\,2^{-t}+A\,2^{-t}=\pi$, an interval on which $\sin\delta_r$
has the sign of $\delta_r$. The decision $P(\ket{0})<\tfrac12\iff\sin\delta_r<0$
therefore coincides with $\varphi_x<\varphi_y$.

For $\varphi(u)=c\,(u-a)$ the phases lie in $[0,W]$ and their differences in
$[-W,W]$, so the band hypothesis reads $W<\pi-A\,2^{-t}$, that is
$W\,2^{t}/(2^{t}-1)<\pi$, which holds for every $t\ge1$ whenever
$W<\tfrac{\pi}{2}$ and is tight at $t=1$. Since $\varphi(a)=0$ and
$A(1-2^{-t})\ge W$, every data phase lies between $0$ and the largest grid
point, so no phase saturates and round-to-nearest moves each of them by at most
$A\,2^{-(t+1)}$. Finally, the encoding scales every gap by $c$, and
$A\,2^{-t}=c\,(b-a)/(2^{t}-1)$, so $|x-y|>(b-a)/(2^{t}-1)$ and
$|\varphi_x-\varphi_y|>A\,2^{-t}$ are the same statement.
\end{proof}

Precision therefore enters the decision only as a resolution limit. When
$|x-y|\le(b-a)/(2^{t}-1)$ the cascade may disagree with the hard-coded circuit;
when $|x-y|>(b-a)/(2^{t}-1)$ the two agree. Each additional bit shrinks the
window by a factor $(2^{t+1}-1)/(2^{t}-1)$, slightly more than two. The worst
case attains the bound, and on a uniform grid the reason is simple: the
encoding is strictly increasing and rounding is monotone, so two grid points
never swap order, and rounding changes a decision only when both phases fall
into the same rounding cell, which happens for gaps up to one full step and
for no larger gap.

\paragraph*{Cost against an exact comparator.}
The width-one count of \Cref{sec:comparator} applies to hard-coded values only.
On two $t$-qubit registers the cascade uses one extra qubit and $2t$
controlled-phase gates, that is $2t{+}1$ qubits at depth $\mathcal{O}(t)$, and
returns a statistical bit. An exact reversible comparator on the same registers
uses the same $2t{+}1$ qubits, has depth $\Theta(\log_2 t)$ and returns an
exact bit~\cite{vandaele2026}. Its $\Theta(t)$ gate count is of the same order
as ours, but over a different gate set: the cascade uses $2t$ controlled-phase
gates with arbitrary angles, which that design avoids because such gates are
expensive to synthesize fault-tolerantly. The two gate counts are therefore not
directly comparable. The cascade differs from an exact comparator in two respects. It leaves the
registers unmeasured, so the operands stay in superposition after the
comparison. And its phases remain adjustable: powering (\Cref{sec:accuracy})
rescales the fixed angles of the cascade without adding gates, and only a
statistical readout benefits from the larger margin.

\section{Margin Amplification by Phase Powering}\label{sec:accuracy}

The readout of \Cref{sec:cost} spends shots quadratically in the inverse
separation, and the whole of that cost comes from a margin that vanishes at
the tie. Since the encoded phases are inputs under our control, we may encode
an integer multiple of each angle instead: scaling both by the same integer
$r$ scales their difference by $r$, and so widens the margin where the readout needs it most.

\begin{proposition}[Phase powering]\label{prop:powering}
Let $r\ge 1$ be an integer, and replace the data gates of \Cref{fig:W} by
$P(r\varphi_x)$ and $P(-r\varphi_y)$, or the fixed angles $\pm\beta_i$ of
\Cref{fig:comparator-cascade} by $\pm r\beta_i$. The circuit produces
\[
P(\ket{0})=\tfrac12+\tfrac12\sin\bigl(r(\varphi_x-\varphi_y)\bigr),
\]
and the decision rule $P(\ket{0})<\tfrac12\iff \varphi_x<\varphi_y$ remains
valid whenever $r\,|\varphi_x-\varphi_y|<\pi$. Within the band
$r\,|\varphi_x-\varphi_y|\le\tfrac{\pi}{2}$ the margin satisfies
$\Delta_r=\tfrac12\bigl|\sin\bigl(r(\varphi_x-\varphi_y)\bigr)\bigr|
 \ge \tfrac{r}{\pi}\,|\varphi_x-\varphi_y|$,
an $r$-fold amplification near ties.
\end{proposition}
 
\begin{proof}[Proof (margin amplification)]
Write $d=\varphi_x-\varphi_y$ for the separation. The proof has three steps:
the closed form, the validity band, and the margin bound.
 
\emph{Closed form.} In both circuits the diagonal gates commute and their
angles add, so the comparator qubit accumulates the single phase
$\phi=-\tfrac{\pi}{2}+r(\varphi_x-\varphi_y)=-\tfrac{\pi}{2}+r\,d$: with
hard-coded values because the data gates carry $r\varphi_x$ and $-r\varphi_y$,
and with the cascade because $\sum_i a_i\,(r\beta_i)=r\sum_i a_i\beta_i$ by
\Cref{lemma:cascade}. This is the setting of \Cref{thm:quac} with
$\varphi_x-\varphi_y$ replaced by $r\,d$, and the interference calculation
there gives $P(\ket{0})=\tfrac12+\tfrac12\sin(r\,d)$.
 
\emph{Validity band.} By the closed form, $P(\ket{0})<\tfrac12$ holds exactly
when $\sin(r\,d)<0$. On the interval $(-\pi,\pi)$ the sine has the sign of its
argument, so as long as $|r\,d|<\pi$ we have $\sin(r\,d)<0$ exactly when
$r\,d<0$, that is, exactly when $\varphi_x<\varphi_y$, because $r>0$. This is
the band $r\,|\varphi_x-\varphi_y|<\pi$ of the statement. Beyond it the sine
changes sign and the rule inverts, so the band cannot be widened.
 
\emph{Margin bound.} Inside the narrower band $|r\,d|\le\tfrac{\pi}{2}$ the
margin is $\Delta_r=\tfrac12|\sin(r\,d)|=\tfrac12\sin|r\,d|$, since the sine is
odd. On $[0,\tfrac{\pi}{2}]$ the sine is concave, vanishes at $0$ and equals
$1$ at $\tfrac{\pi}{2}$, so it lies above the chord joining those two points:
$\sin\theta\ge\tfrac{2}{\pi}\,\theta$ for every $\theta\in[0,\tfrac{\pi}{2}]$.
Applied at $\theta=|r\,d|$ this gives
\[
\Delta_r
\;\ge\;\tfrac12\cdot\tfrac{2}{\pi}\,|r\,d|
\;=\;\tfrac{r}{\pi}\,|\varphi_x-\varphi_y| .
\]
The unpowered margin behaves as $\Delta_1=\tfrac12|\sin d|\approx\tfrac12|d|$
near a tie, whereas the powered one is at least $\tfrac{r}{\pi}|d|$ and behaves
as $\tfrac12 r|d|$ there, so the amplification near ties is by the factor $r$.
\end{proof}

Powering needs no multiplier. With hard-coded values we write
$r\varphi_x \bmod 2\pi$ into the phase gate. With the register-driven cascade
the registers are never touched: they hold the phases, as the bits
$a_0\cdots a_{t-1}$ of \Cref{lemma:cascade}, and the phase acts on the
comparator qubit only when the cascade's fixed angles act on those bits. The
cascade contributes $\theta=\sum_i a_i\beta_i$, so
$r\theta=\sum_i a_i\,(r\beta_i)$, and the factor $r$ moves into the fixed
angles. The layout of \Cref{fig:comparator-cascade} is unchanged; only its
angles become $r\beta_i \bmod 2\pi$. Multiplying the register content by $r$
and then running the original cascade would produce the same phase, at the
cost of a multiplier and extra qubits.
 
Powering does not enlarge the rounding window. With angles $r\beta_i$ the
registers still hold phases rounded on the grid of step $A/2^{t}$, and the
cascade multiplies the rounding error and the separation by the same $r$, so
the proof of \Cref{prop:window} goes through with $d$ replaced by $r\,d$: the
two circuits agree whenever $|x-y|>(b-a)/(2^{t}-1)$, exactly as for $r=1$,
provided the band hypothesis holds in its powered form
$r\,|\varphi_x-\varphi_y|<\pi-r\,A\,2^{-t}$. Under the standing convention
$W<\tfrac{\pi}{2}$ this reads $r\,W\,2^{t}/(2^{t}-1)<\pi$, which is the
admissibility condition on $r$ that \Cref{sec:adaptive} uses.

\subsection{The Adaptive Schedule}\label{sec:adaptive}
The margin of the comparator is $\tfrac12|\sin u|$, where
$u=|\varphi_x-\varphi_y|$ is the separation of the pair. For a narrow pair,
say $u=0.05$, the margin is about $0.025$, and telling $P(\ket{0})=0.525$ from
$0.5$ takes thousands of shots. Powering multiplies $u$ by $r$ and widens the
margin, but \Cref{prop:powering} holds only while $r\,u<\pi$, and past that
point the circuit returns the opposite order with
a wide margin, since the sine changes sign at $\pi$. For $u=0.8$, for example,
$r=3$ gives a powered separation of $2.4$ and
$P(\ket{0})=\tfrac12(1+\sin 2.4)\approx0.84$, a clear vote for the correct
order, whereas $r=5$ gives $4.0$ and $P(\ket{0})\approx0.12$, an equally clear
vote for the wrong one, and the readout offers no way to tell the two apart.
The right $r$ for a pair is the one that keeps $r\,u$ well inside $(0,\pi)$,
so it depends on $u$; but $u$ is the quantity the comparator does not know,
and a procedure that knew it would already know the order. A fixed $r$ settles
for the bound that the data range provides: under the encoding of
\Cref{sec:encodings} every separation is at most $W$, so any $r<\pi/W$ is safe
on the whole range, but that $r$ is calibrated on the widest pair and stays
small for the narrow pairs that need it most.

The adaptive schedule finds a suitable $r$ for the pair at hand without
knowing $u$, by approaching the limit from below.
\begin{enumerate}
\item \emph{Level $0$.} Run the unpowered circuit ($r=1$) for $n_0$ shots and
compute the frequency $\widehat{p}_0$ of the outcome $\ket{0}$. The
\emph{empirical margin} is $|\widehat{p}_0-\tfrac12|$.
\item \emph{Declare or double.} If the empirical margin reaches a threshold
$\tau$, declare the order, $\varphi_x<\varphi_y$ when
$\widehat{p}_0<\tfrac12$ and $\varphi_x>\varphi_y$ otherwise, and stop. If it
does not, move to the next level with $r$ doubled: $r=2$, then $4$, then $8$,
and so on. Each level costs the same $n_0$ shots and uses the same circuit,
since powering only rewrites the fixed angles.
\item \emph{Give up at the budget.} If no level has declared after $L$
levels, report a tie at resolution $\pi2^{-L}$, because any separation of at
least that size would have crossed the threshold within the budget.
\end{enumerate}

The schedule works because of a window. When the powered separation $r\,u$
lies in $[\tfrac{\pi}{4},\tfrac{\pi}{2}]$, the margin is at least
$\tfrac12\sin\tfrac{\pi}{4}\approx0.35$, which is large, and $r\,u$ is still
far from $\pi$, so the sign is correct. The two endpoints of the window differ
by a factor of two, and each level multiplies $r\,u$ by exactly two, so a
separation that starts below $\tfrac{\pi}{4}$ cannot jump over the window and
some level lands inside it. That level is
$\ell^{\ast}=\lceil\log_2(\pi/4u)\rceil$, and the schedule stops there at the
latest. It also does not fail on the way: every level before $\ell^{\ast}$ has
$r\,u<\tfrac{\pi}{4}$, so none of them has wrapped, and any of them that
declares declares correctly. The wrap region $r\,u\ge\pi$ is therefore never
reached, which is exactly the hazard that a fixed $r$ has to guard against in
advance.

\Cref{tab:schedule-example} traces the schedule on $u=0.05$ at confidence
$0.95$. With $L=8$, $\delta_0=0.05$, the smallest admissible $n_0=93$ and the midpoint threshold $\tau=\tfrac14\sin\tfrac{\pi}{4}\approx0.177$, it declares at level $3$, after four levels of $93$ shots each, so $372$ shots in total, whereas the unpowered readout of \Cref{eq:shots} needs about $2\,950$ shots for the same
decision, a factor of about eight. Level $\ell^{\ast}=4$ is where the schedule
stops \emph{at the latest}, not where it stops here: declaring earlier is safe,
because no level before $\ell^{\ast}$ has wrapped. The number of
levels is $\log_2(1/u)$ up to a constant, and that is where the logarithm in
the shot count of \Cref{cor:adaptive} comes from.

\begin{table}[t]
\centering
\caption{The adaptive schedule on a pair with separation $u=0.05$ radians.
Each level multiplies the powered separation by two; the first level whose
true margin exceeds the threshold lies inside the window
$[\tfrac{\pi}{4},\tfrac{\pi}{2}]$.}
\label{tab:schedule-example}
\small
\begin{tabular}{ccccl}
\toprule
\textbf{Level} $\ell$ & $r_\ell$ & $r_\ell u$ & \textbf{Margin} $\tfrac12\sin(r_\ell u)$ & \textbf{Outcome} \\
\midrule
0 & 1  & 0.05 & 0.025 & below threshold \\
1 & 2  & 0.10 & 0.050 & below threshold \\
2 & 4  & 0.20 & 0.099 & below threshold \\
3 & 8  & 0.40 & 0.195 & below threshold \\
4 & 16 & 0.80 & 0.359 & inside the window \\
\bottomrule
\end{tabular}
\end{table}

Two requirements fix the constants. Hoeffding's inequality with a union bound
over the $L$ levels gives a per-level radius
$\varepsilon_0=\sqrt{\ln(2L/\delta_0)/2n_0}$, so that with probability at
least $1-\delta_0$ every level estimates its margin to within $\varepsilon_0$.
The threshold must then lie between two limits: above $\varepsilon_0$, so that
a declaration certifies a nonzero true margin, and at most
$\tfrac12\sin\tfrac{\pi}{4}-\varepsilon_0$, so that the level inside the
window does declare. Both hold at once exactly when $n_0>16\ln(2L/\delta_0)$,
which is about $92$ for $L=8$ and $\delta_0=0.05$. The schedule needs nothing
else in advance: only $u\le\tfrac{3\pi}{4}$, which the standing convention
$W<\tfrac{\pi}{2}$ guarantees for every pair, and the budget $L$, which sets
the finest resolution. The doubling itself comes from Kitaev's phase
estimation~\cite{kitaev1995}, with the difference that it stops at a decision
instead of completing an estimate.

\begin{corollary}[Adaptive powered readout]\label{cor:adaptive}
Fix a confidence $1-\delta_0$ and a level budget $L$. For
$\ell=0,1,\dots,L-1$, run the comparator with powering $r_\ell=2^{\ell}$ for
$n_0=\mathcal{O}(\log_2(L/\delta_0))$ shots, and stop at the first level whose
empirical margin $|\widehat{p}_0-\tfrac12|$ reaches a fixed threshold $\tau$,
declaring the sign of $\widehat{p}_0-\tfrac12$. Let the separation
$u=|\varphi_x-\varphi_y|$ satisfy $u\le\tfrac{3\pi}{4}$ and $2^{L}u\ge\pi$.
Then the declared order is correct with probability at least $1-\delta_0$, and
the schedule uses
$\mathcal{O}\!\big(\big(1+\log_2\tfrac{1}{u}\big)\log_2\tfrac{L}{\delta_0}\big)$
shots in total. If no level declares, the pair is a tie at resolution
$\pi2^{-L}$.
\end{corollary}

\begin{proof}[Proof sketch]
Let $\varepsilon_0$ be the Hoeffding radius of one level and choose
$\varepsilon_0<\tau\le\tfrac12\sin\tfrac{\pi}{4}-\varepsilon_0$. A union bound
over the $L$ levels puts every empirical margin within $\varepsilon_0$ of the
true one with probability at least $1-\delta_0$; on that event a declaration
certifies a true margin of at least $\tau-\varepsilon_0>0$, whose sign is the
true order as long as the declaring level has not wrapped
(\Cref{prop:powering}). If $u\ge\tfrac{\pi}{4}$, level $0$ has true margin
$\tfrac12\sin u\ge\tfrac12\sin\tfrac{\pi}{4}$ and declares; here
$u\le\tfrac{3\pi}{4}$ is needed, since for $u\in(\tfrac{3\pi}{4},\pi)$ the
level-$0$ margin falls below the threshold while every later level wraps. If
$u<\tfrac{\pi}{4}$, the first level with
$2^{\ell}u\in[\tfrac{\pi}{4},\tfrac{\pi}{2}]$ is
$\ell^{\ast}=\lceil\log_2(\pi/4u)\rceil$, every earlier level has not wrapped,
and $\ell^{\ast}$ declares; the hypothesis $2^{L}u\ge\pi$ gives
$\ell^{\ast}\le L-2$, so that level is run. The shot count is
$n_0(\ell^{\ast}+1)=\mathcal{O}\big((1+\log_2(1/u))\,n_0\big)$.
\end{proof}

On a superposition of operand pairs the factor $r$ is a single setting of the
circuit's fixed angles, so every branch receives the same $r$. A fixed $r$
therefore works branch by branch exactly as above: every separation is at most
$W$, so $r<\pi/W$ keeps every branch inside the band, and this is the only
regime available when the state is not measured but consumed coherently
(\Cref{sec:register}). The adaptive schedule instead requires a readout that
measures the registers, or an index register entangled with them, alongside
the comparator qubit. Each shot then reports which pair it sampled, the
empirical margin of each branch is estimated from its own shots, and the
declaration rule of \Cref{cor:adaptive} applies to each branch separately: a
wide pair declares at level $0$ and ignores the later levels at which it would
wrap, a narrow pair continues to its own $\ell^{\ast}$, and the proof goes
through on every branch since $u\le\tfrac{3\pi}{4}$ holds for all of them.
The cost grows in two ways. Each branch needs its $n_0$ shots per level, so
the total scales with the inverse of the smallest branch weight, and holding
every branch correct at once extends the union bound over the $K$ branches,
which raises $n_0$ by an additive $\log_2 K$.

\section{Experiments and Application}\label{sec:exp-setup}

We implement all circuits in Qiskit~\numEnvQiskit{} (Aer~\numEnvAer,
Python~\numEnvPython) and validate them by noiseless simulation: statevector
simulation for the exactness checks, and seeded shot sampling for the readout
comparison and the thresholding run. On a $\numVerGrid\times\numVerGrid$ grid
of inputs $x,y\in[0,1]$ at exact angles, the marginal of the comparator qubit
agrees with the closed form~\eqref{eq:p0} to machine precision, deviating by at
most $\numVerFormulaErr$, and the decision rule reproduces the classical
predicate $x<y$ at every grid point off the diagonal; it errs only at exact
equality, where the margin is zero by construction. When the phases lie on the
$t$-bit grid, the register-driven cascade of \Cref{fig:comparator-cascade}
reproduces the exact-angle circuit to $\numVerCascadeErr$. 

\subsection{Order Decision versus Value Estimation}\label{sec:exp-order}
\begin{figure}[t]
\centering
\includegraphics[width=0.7\linewidth]{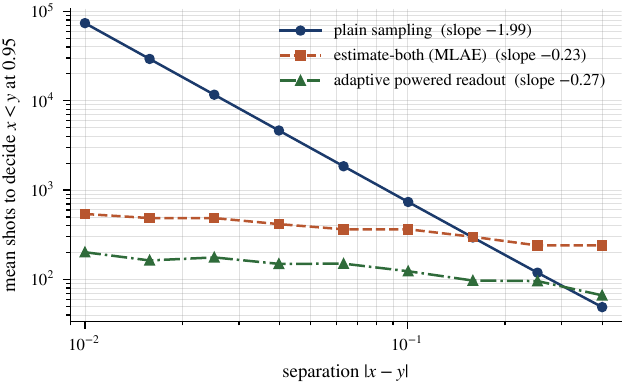}
\caption{Shots to decide $x<y$ at confidence $0.95$ versus the separation
$|x-y|$, on log--log axes. All three series count the same unit: a
\emph{shot} is one run of a circuit ending in one measurement. Plain sampling
of the comparator grows like $|x-y|^{-2}$, the exponent of the shot
bound; estimating both angles by amplitude estimation is
flatter in shots but pays for it in circuit depth, which the text reports
separately; the adaptive powered readout on register phases is near-flat, its logarithmic count.}

\Description{Log-log plot of the shots needed to decide the order versus the
separation between the two values, for three series in one unit: plain sampling
of the comparator, falling steeply with slope about minus two; estimating both
values by amplitude estimation, on a flatter curve above; and the powered
register readout, near-flat and lowest except at the widest separation, where
plain sampling crosses below it.}
\label{fig:order-vs-estimate}
\end{figure}
The decision rule uses only the sign of $\varphi_x-\varphi_y$, and the
baseline should be judged against that. The estimate-both baseline recovers
both values and compares them, and recovering one value to additive error
$\varepsilon$ by amplitude estimation costs $\Theta(1/\varepsilon)$
applications of the state-preparation unitary~\cite{brassard2002}, the
Heisenberg limit. The baseline lives in its own input model: it is the cost of
reaching the same decision when the values originate as amplitudes of a
prepared state, a regime that our construction does not cover, and it is the
only row of \Cref{tab:resources} whose operands are amplitudes. Near a tie the
error has to shrink with the separation, so this route spends
$\Theta(1/|\varphi_x-\varphi_y|)$ applications per value and returns two
numbers where one bit was asked for. Deciding a sign should therefore cost
less than estimating two values, and \Cref{sec:accuracy} shows that it does:
the adaptive powered readout decides in
$\mathcal{O}(1+\log_2(1/|\varphi_x-\varphi_y|))$ shots at fixed confidence
and level budget $L$. Here we measure the two protocols against each
other.

\Cref{fig:order-vs-estimate} reports the mean number of shots needed to decide
$x<y$ at confidence $0.95$ as a function of the separation $|x-y|$, for three
protocols: the comparator with plain sampling, the comparator with the adaptive
powered readout on register phases, and the estimate-both baseline. All three
series count the same unit, shots. The baseline is given the target error $\varepsilon=|x-y|/2$ at each separation, the coarsest error at which two independent estimates still separate the pair, and runs maximum-likelihood amplitude estimation on the incremental power ladder $m_k=0,1,2,4,\dots,2^{K-1}$ with $K$ the smallest integer such that $2^{K-1}\ge1/\varepsilon$, spending the same $n=\lceil\ln(2(K+1)/\delta)/(2(1/3)^{2})\rceil$ shots at every power and repeating the whole ladder for each of the two values. At $|x-y|=0.01$ this is $K=9$, ten circuits per value and $n=27$ shots per circuit.

The three curves follow the analysis. Plain sampling tracks the shot
bound~\Cref{eq:shots}: its log-log slope is $\numOveSlopePlain$, which matches
$\mathcal{O}(1/|x-y|^{2})$, and it climbs from about fifty shots at the widest
separation to tens of thousands at $|x-y|=0.01$. The estimate-both baseline,
which recovers each value by amplitude estimation~\cite{brassard2002}, is much
flatter in shots, with a log-log slope of $\numOveSlopeEstShots$; its
Heisenberg $\mathcal{O}(1/\varepsilon)$ behaviour does not show in the shot
count but in the depth of the circuits that those shots run, which we price
below. The adaptive powered readout (\Cref{cor:adaptive}) is nearly flat, with
a slope of $\numOveSlopeReg$, as a shot count logarithmic in $1/|x-y|$; the curve is not monotone, because the smallest admissible $t$ rounds each pair onto its own
grid and the rounded separation can exceed or fall short of the nominal one. At the tightest separation, $|x-y|=0.01$, it decides in
$\numOveRegShots$ shots, against $\numOveEstShots$ for the baseline
($\numOveRatioEstShots$ times as many) and $\numOvePlainShots$ for plain
sampling (about $\numOveRatioPlain$ times as many). At the widest separations
the order reverses: above $|x-y|\approx0.3$ plain sampling needs fewer shots
than the schedule, because the schedule spends a full level of shots even when
level $0$ would have declared with fewer, and the baseline is the most
expensive of the three. Every protocol keeps its empirical accuracy at or above
$\numOveAccFloor$, above the design confidence of $\numOveAccTarget$.

This run supports \Cref{cor:adaptive} under a weaker rule than the corollary
assumes. The level budget is $L=9$, the smallest budget meeting the corollary's
second hypothesis $2^{L}u\ge\pi$ at the tightest separation $u=0.01$ (at $L=8$
one has $2^{8}u=2.56<\pi$, and the resolution $\pi2^{-8}\approx0.012$ would
report that pair as a tie), and each level spends $n_0=25$ shots, fewer than
the $n_0>16\ln(2L/\delta_0)=94.2$ that the corollary requires. The schedule
declares as soon as the empirical margin exceeds the Hoeffding radius
$\varepsilon_0=\sqrt{\ln(2L/\delta_0)/2n_0}=0.343$ itself, that is with
$\tau=\varepsilon_0$, so it is a different decision rule, not a cheaper
instance of the corollary's, and its confidence is the empirical floor quoted
above rather than a proven bound. Each sweep point uses the smallest $t$ with
$(b-a)/(2^{t}-1)<|x-y|$, as \Cref{prop:window} prescribes. What the corollary contributes to the
run is the guarantee that no level overshoots the validity band, which is what
lets the schedule run without knowing the separation.

Depth separates the protocols more sharply than shots do. Powering rewrites
the fixed angles of the cascade and adds no gate (\Cref{prop:powering}), so no
level is deeper than level $0$, the cascade of \Cref{fig:comparator-cascade},
of depth $\mathcal{O}(t)$; and since \Cref{prop:window} ties the precision to
the separation through $(b-a)/(2^{t}-1)<|x-y|$, the depth of the comparator
grows like $\log_2(1/|x-y|)$. The deepest circuit of the baseline applies the
Grover operator $\numOveEstMaxPower$ times at $|x-y|=0.01$, in line with the
$\Theta(1/\varepsilon)$ depth of amplitude estimation. Logarithmic against
linear in the inverse separation is an exponential gap in depth, between
protocols that admittedly answer different questions, an order against two
values. Where the coherence budget bounds depth rather than repetitions, that
gap is the one that matters.

\subsection{Thresholding an Image Held in Superposition}\label{sec:app-threshold}
\begin{figure}[t]
\centering
\includegraphics[width=0.7\linewidth]{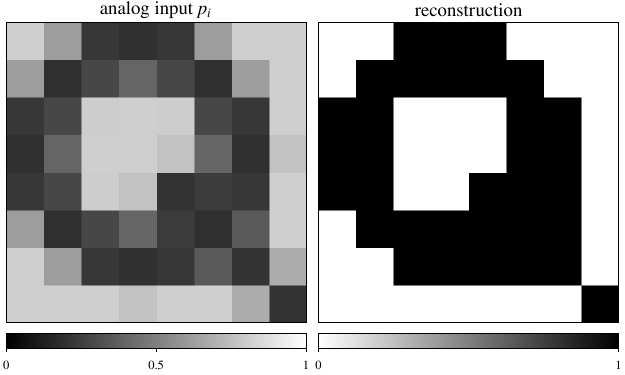}
\caption{Thresholding an image held in superposition. Left: analog input. Center: empirical per-shot
decision accuracy per pixel. Right: binarization reconstructed from the quantum
samples, identical to the classical reference.}
\Description{Three eight-by-eight images side by side: an analog gray-level
input image showing the letter Q, a map of the per-pixel decision accuracy,
and the reconstructed binary image, identical to the classical reference.}
\label{fig:binarization}
\end{figure}

Binarizing an image against a constant threshold is a standard application of
quantum comparators in image processing~\cite{xia2019}. Let the image state be
$\tfrac{1}{\sqrt{M}}\sum_{i}\ket{i}^{6}\ket{\varphi(p_i)}^{t}$ with
$\varphi(p_i)=c\,(p_i-a)$ on $[a,b]=[0,1]$ and $M=64$ pixels: the $8\times8$ gray-level rendering of
the letter Q on the left of \Cref{fig:binarization}, whose levels lie in
$[\numThrLevelLo,\numThrLevelHi]$. The asymmetry of the operands suits the
comparator. The threshold $p_{\mathrm{thr}}=0.5$, with phase
$\varphi(p_{\mathrm{thr}})$, is a single classical parameter: we hard-code it
into the fixed phase gates, and it never enters a register, whereas the pixels
arrive through a cascade. The $t$ fixed-angle controlled-phase gates that
compare one pixel therefore compare all $64$ at once.

After the comparator runs, the state is a superposition over pixels in which
every branch carries its own decision: the comparator qubit of branch $i$ is a
coin biased towards $\ket{1}$ when $p_i<p_{\mathrm{thr}}$ and towards $\ket{0}$
otherwise, with a bias that grows with the powered separation
$r\,|\varphi(p_i)-\varphi(p_{\mathrm{thr}})|$. The algorithm need not measure
at this point. The thresholded pairs remain in superposition, entangled with
the index and pixel registers, and a downstream routine can consume them
coherently: amplitude amplification on the branches whose comparator qubit
reads $\ket{1}$ raises the weight of the pixels below threshold, and a
counting routine on the same flag estimates how many there are.

The demonstration reported here measures instead, because its purpose is to
verify the decisions against the classical binarization. A single comparator
qubit gives one biased coin, so we replicate the readout: we apply the
comparator $k$ times, each application writing into a fresh qubit, and each
shot measures the index register together with those $k$ qubits, whose
majority is the decision for the pixel that the index names. We reduce the
votes classically and need no result qubit.

We take $t=7$, $k=9$ and powering degree $r=6$. The bound $r<\pi/W$ of
\Cref{sec:adaptive} is calibrated on the whole data range and would admit only
$r\le\numPowRmax$ here; but it is the separations the image actually contains
that have to stay inside the band, and the largest of those is $\numThrSep$, so
every $r<\numThrRmax$ is admissible and the first wrapping degree is
$r=\numThrRwrap$, at a powered separation of $\numThrWrap>\pi$. Every powered
separation $r\,|\varphi(p_i)-\varphi(p_{\mathrm{thr}})|$ then stays below
$\numThrPowSep<\pi$, inside the validity band of \Cref{prop:powering}. The
circuit uses $6+t+k=22$ qubits.

\section{Conclusion}\label{sec:conclusion}
This paper presents the Quantum Phase-based Comparator, a phase interferometer
that decides the order of two values encoded as phases, together with two
circuits that implement it, one with hard-coded values and one driven by
registers. The two phases enter a single qubit between two Hadamard gates, a
fixed offset of $-\tfrac{\pi}{2}$ places the center at probability one half, and
the side of one half on which the probability of $\ket{0}$ falls gives the
order. Any strictly increasing map into a phase range narrower than $\pi$
serves as encoding. The two circuits differ only in how
the phases reach the qubit. With hard-coded values the comparison costs one
qubit and five gates. With the register-driven cascade it compares values held
in two $t$-qubit registers at the cost of one qubit beyond those registers and
depth $\mathcal{O}(t)$, and rounding can change the decision only inside a
window of width $(b-a)/(2^{t}-1)$. Because the cascade acts on each basis state of the registers separately, one circuit compares a whole superposition of pairs, and the decisions stay in superposition for a downstream quantum routine.

The comparator returns a biased coin, so resolving close values costs shots,
quadratically in the inverse separation. Because the algorithm sets the phases itself, in the gate parameters or in the
fixed angles of the cascade, scaling both by an integer widens the margin
without adding a gate or a qubit, and the adaptive doubling schedule turns
this amplification into a shot count logarithmic in the inverse separation.

To the best of our knowledge, no earlier quantum circuit decides an order
relation by interfering two independently encoded phases on a single qubit.
The ingredients themselves are not new: the interferometer between two
Hadamard gates is standard~\cite{ekert2002}, and the cascade of controlled
rotations comes from the state-preparation
literature~\cite{mottonen2005,mitarai2019adc}. What is new is their
composition, and in particular the fixed offset that centers the decision at
equality and so turns an interferometer into a comparator.

Three directions follow. The first is a downstream algorithm that takes the
compared superposition as input without measuring it. Amplitude amplification
over the branches whose comparator qubit reads $\ket{1}$ is the simplest
example, and the same flag serves any routine that needs only aggregate
quantities, which amplitude estimation can extract:
norm, trace and variance
estimation~\cite{cade2017quantum,poggiali2026more,bernasconi2024quantum,poggiali2023quantum}
are representative instances, and the comparator pays off there because the
thresholding step leaves the state coherent for them. The second is noise:
every result here is exact or noiselessly simulated, and the margin of
\Cref{sec:cost} is the natural quantity through which depolarizing and readout
errors enter, since both shrink it. The third is the encoding: a value held as
the amplitude of a state can reach the same interferometer through amplitude
estimation and an analog-to-digital conversion~\cite{mitarai2019adc}, a route
that we leave open.

\begin{acks}
The authors used a large language model (Claude, Anthropic) to assist with
drafting and revising the manuscript text.
\end{acks}

\bibliographystyle{ACM-Reference-Format}
\bibliography{references}

@inproceedings{thapliyal2010,
  author    = {Thapliyal, Himanshu and Ranganathan, Nagarajan and Ferreira, Ryan},
  title     = {Design of a Comparator Tree Based on Reversible Logic},
  booktitle = {Proc. 10th IEEE Int. Conf. Nanotechnology (IEEE-NANO)},
  year      = {2010},
  pages     = {1113--1116},
  publisher = {IEEE},
  doi       = {10.1109/NANO.2010.5697872}
}

@inproceedings{vudadha2012,
  author    = {Vudadha, Chetan and Phaneendra, P. Sai and Sreehari, V. and Ahmed, Syed Ershad and Muthukrishnan, N. Moorthy and Srinivas, M. B.},
  title     = {Design of Prefix-Based Optimal Reversible Comparator},
  booktitle = {Proc. IEEE Computer Society Annu. Symp. VLSI (ISVLSI)},
  year      = {2012},
  pages     = {201--206},
  publisher = {IEEE},
  doi       = {10.1109/ISVLSI.2012.49}
}

@article{cade2017quantum,
  title={The quantum complexity of computing Schatten $ p $-norms},
  author={Cade, Chris and Montanaro, Ashley},
  journal={arXiv preprint arXiv:1706.09279},
  year={2017}
}

@inproceedings{poggiali2023quantum,
  title={Quantum Feature Selection with Variance Estimation.},
  author={Poggiali, Alessandro and Bernasconi, Anna and Berti, Alessandro and Del Corso, Gianna M and Guidotti, Riccardo and others},
  booktitle={ESANN},
  year={2023}
}

@article{bernasconi2024quantum,
  title={Quantum subroutine for variance estimation: algorithmic design and applications},
  author={Bernasconi, Anna and Berti, Alessandro and Del Corso, Gianna M and Guidotti, Riccardo and Poggiali, Alessandro},
  journal={Quantum Machine Intelligence},
  volume={6},
  number={2},
  pages={78},
  year={2024},
  publisher={Springer}
}

@article{poggiali2026more,
  title={A more efficient quantum circuit for estimating the variance},
  author={Poggiali, Alessandro and Ju, Jiwon},
  journal={Quantum Machine Intelligence},
  volume={8},
  number={1},
  pages={34},
  year={2026},
  publisher={Springer}
}

@misc{qiskit2024,
      title={Quantum computing with {Q}iskit},
      author={Javadi-Abhari, Ali and Treinish, Matthew and Krsulich, Kevin and Wood, Christopher J. and Lishman, Jake and Gacon, Julien and Martiel, Simon and Nation, Paul D. and Bishop, Lev S. and Cross, Andrew W. and Johnson, Blake R. and Gambetta, Jay M.},
      year={2024},
      doi={10.48550/arXiv.2405.08810},
      eprint={2405.08810},
      archivePrefix={arXiv},
      primaryClass={quant-ph}
}

@article{brassard2000quantum,
  title={Quantum amplitude amplification and estimation},
  author={Brassard, Gilles and Hoyer, Peter and Mosca, Michele and Tapp, Alain},
  journal={arXiv preprint quant-ph/0005055},
  year={2000}
}

@article{oliveira2007,
  author  = {Oliveira, David Sena and Ramos, Rubens Viana},
  title   = {Quantum Bit String Comparator: Circuits and Applications},
  journal = {Quantum Computers and Computing},
  volume  = {7},
  number  = {1},
  pages   = {17--26},
  year    = {2007}
}

@misc{shahzad2023,
  author       = {Shahzad, Khuram and Khan, Omar Usman},
  title        = {A Generalized Space-Efficient Algorithm for Quantum Bit String Comparators},
  year         = {2023},
  eprint       = {2311.06573},
  archivePrefix = {arXiv},
  primaryClass = {quant-ph},
  howpublished = {arXiv:2311.06573},
  doi          = {10.48550/arXiv.2311.06573}
}

@article{yuan2023,
  author  = {Yuan, Yewei and Wang, Chao and Wang, Bei and Chen, Zhao-Yun and Dou, Meng-Han and Wu, Yu-Chun and Guo, Guo-Ping},
  title   = {An improved {QFT}-based quantum comparator and extended modular arithmetic using one ancilla qubit},
  journal = {New Journal of Physics},
  volume  = {25},
  number  = {10},
  pages   = {103011},
  year    = {2023},
  doi     = {10.1088/1367-2630/acfd52}
}

@article{sahin2020,
  author  = {{\c{S}}ahin, Engin},
  title   = {Quantum arithmetic operations based on quantum {Fourier} transform on signed integers},
  journal = {International Journal of Quantum Information},
  volume  = {18},
  number  = {6},
  pages   = {2050035},
  year    = {2020},
  doi     = {10.1142/S0219749920500355}
}

@article{xia2019,
  author  = {Xia, Haiying and Li, Haisheng and Zhang, Han and Liang, Yan and Xin, Jing},
  title   = {Novel multi-bit quantum comparators and their application in image binarization},
  journal = {Quantum Information Processing},
  volume  = {18},
  number  = {7},
  pages   = {229},
  year    = {2019},
  doi     = {10.1007/s11128-019-2334-2}
}

@misc{durr1996,
  author       = {D{\"u}rr, Christoph and H{\o}yer, Peter},
  title        = {A Quantum Algorithm for Finding the Minimum},
  year         = {1996},
  eprint       = {quant-ph/9607014},
  archivePrefix = {arXiv},
  primaryClass = {quant-ph},
  howpublished = {arXiv:quant-ph/9607014}
}

@article{mottonen2005,
  author  = {M{\"o}tt{\"o}nen, Mikko and Vartiainen, Juha J. and Bergholm, Ville and Salomaa, Martti M.},
  title   = {Transformation of quantum states using uniformly controlled rotations},
  journal = {Quantum Information and Computation},
  volume  = {5},
  number  = {6},
  pages   = {467--473},
  year    = {2005},
  eprint  = {quant-ph/0407010},
  archivePrefix = {arXiv},
  primaryClass = {quant-ph}
}

@article{mitarai2019adc,
  author  = {Kosuke Mitarai and Masahiro Kitagawa and Keisuke Fujii},
  title   = {Quantum Analog-Digital Conversion},
  journal = {Physical Review A},
  volume  = {99},
  number  = {1},
  pages   = {012301},
  year    = {2019},
  doi     = {10.1103/PhysRevA.99.012301}
}

@article{ekert2002,
  author    = {Ekert, Artur K. and Alves, Carolina Moura and Oi, Daniel K. L. and Horodecki, Micha{\l} and Horodecki, Pawe{\l} and Kwek, L. C.},
  title     = {Direct Estimations of Linear and Nonlinear Functionals of a Quantum State},
  journal   = {Physical Review Letters},
  volume    = {88},
  number    = {21},
  pages     = {217901},
  year      = {2002},
  doi       = {10.1103/PhysRevLett.88.217901}
}

@misc{ohno2026,
  author       = {Ohno, Hiroshi},
  title        = {Approximate Cosine Similarity Estimation via an Angle-Encoding Hadamard Test},
  year         = {2026},
  eprint       = {2604.15867},
  archivePrefix = {arXiv},
  primaryClass = {quant-ph},
  howpublished = {arXiv:2604.15867}
}

@inproceedings{park2021,
  author    = {Park, Daniel K. and Blank, Carsten and Petruccione, Francesco},
  title     = {Robust Quantum Classifier with Minimal Overhead},
  booktitle = {2021 International Joint Conference on Neural Networks (IJCNN)},
  pages     = {1--7},
  year      = {2021},
  doi       = {10.1109/IJCNN52387.2021.9533403}
}

@article{larose2020,
  author    = {LaRose, Ryan and Coyle, Brian},
  title     = {Robust Data Encodings for Quantum Classifiers},
  journal   = {Physical Review A},
  volume    = {102},
  number    = {3},
  pages     = {032420},
  year      = {2020},
  doi       = {10.1103/PhysRevA.102.032420}
}

@article{schuld2017,
  author  = {Schuld, Maria and Fingerhuth, Mark and Petruccione, Francesco},
  title   = {Implementing a distance-based classifier with a quantum interference circuit},
  journal = {EPL (Europhysics Letters)},
  volume  = {119},
  number  = {6},
  pages   = {60002},
  year    = {2017},
  doi     = {10.1209/0295-5075/119/60002},
  eprint  = {1703.10793},
  archivePrefix = {arXiv},
  primaryClass  = {quant-ph}
}

@misc{ahuja1999,
  author        = {Ahuja, Ashish and Kapoor, Sanjiv},
  title         = {A Quantum Algorithm for finding the Maximum},
  year          = {1999},
  eprint        = {quant-ph/9911082},
  archiveprefix = {arXiv},
  doi           = {10.48550/arXiv.quant-ph/9911082},
  note          = {arXiv:quant-ph/9911082}
}

@misc{kitaev1995,
  author        = {A. Yu. Kitaev},
  title         = {Quantum measurements and the {A}belian {S}tabilizer {P}roblem},
  year          = {1995},
  month         = nov,
  howpublished  = {arXiv:quant-ph/9511026},
  eprint        = {quant-ph/9511026},
  archiveprefix = {arXiv},
  doi           = {10.48550/arXiv.quant-ph/9511026}
}

@article{suzuki2020,
  author        = {Suzuki, Yohichi and Uno, Shumpei and Raymond, Rudy and Tanaka, Tomoki and Onodera, Tamiya and Yamamoto, Naoki},
  title         = {Amplitude estimation without phase estimation},
  journal       = {Quantum Information Processing},
  volume        = {19},
  number        = {2},
  pages         = {75},
  year          = {2020},
  doi           = {10.1007/s11128-019-2565-2},
  eprint        = {1904.10246},
  archivePrefix = {arXiv}
}

@article{grinko2021,
  author        = {Grinko, Dmitry and Gacon, Julien and Zoufal, Christa and Woerner, Stefan},
  title         = {Iterative quantum amplitude estimation},
  journal       = {npj Quantum Information},
  volume        = {7},
  number        = {1},
  pages         = {52},
  year          = {2021},
  doi           = {10.1038/s41534-021-00379-1},
  eprint        = {1912.05559},
  archiveprefix = {arXiv}
}

@incollection{brassard2002,
  author        = {Brassard, Gilles and H{\o}yer, Peter and Mosca, Michele and Tapp, Alain},
  title         = {Quantum Amplitude Amplification and Estimation},
  booktitle     = {Quantum Computation and Information},
  editor        = {Lomonaco, Jr., Samuel J. and Brandt, Howard E.},
  series        = {Contemporary Mathematics},
  volume        = {305},
  pages         = {53--74},
  publisher     = {American Mathematical Society},
  year          = {2002},
  doi           = {10.1090/conm/305/05215},
  eprint        = {quant-ph/0005055},
  archiveprefix = {arXiv}
}

@book{nielsen2010,
  author    = {Michael A. Nielsen and Isaac L. Chuang},
  title     = {Quantum Computation and Quantum Information: 10th Anniversary Edition},
  publisher = {Cambridge University Press},
  address   = {Cambridge, UK},
  year      = {2010},
  doi       = {10.1017/CBO9780511976667},
  isbn      = {978-1-107-00217-3}
}

@misc{berti2025qram,
  author        = {Berti, Alessandro and Ghisoni, Francesco},
  title         = {Efficient Quantum State Preparation with Bucket Brigade {QRAM}},
  year          = {2025},
  eprint        = {2510.16149},
  archivePrefix = {arXiv},
  primaryClass  = {quant-ph},
  doi           = {10.48550/arXiv.2510.16149},
  howpublished  = {arXiv:2510.16149 [quant-ph]}
}

@misc{vandaele2026,
  author        = {Vandaele, Vivien},
  title         = {Asymptotically Optimal Quantum Circuits for Comparators and Incrementers},
  year          = {2026},
  eprint        = {2603.12917},
  archivePrefix = {arXiv},
  primaryClass  = {quant-ph},
  doi           = {10.48550/arXiv.2603.12917},
  howpublished  = {arXiv:2603.12917 [quant-ph]}
}

\end{document}